\documentclass[10pt,conference]{IEEEtran}
\usepackage{tikz,forest,graphicx,times,amssymb,amsmath,bbm,cite,amsthm,float}
\usepackage{algorithm,enumerate,multirow,caption}
\usepackage{sidecap}
\usepackage{verbatim}
\usepackage{textcomp}
\usetikzlibrary{shapes,arrows}

\usepackage[top=0.75in,bottom=0.75in,left=0.75in,right=0.75in]{geometry}
\allowdisplaybreaks

\newtheorem{theorem}{Theorem}[section]

\newtheorem{proposition}[theorem]{Proposition}
\newtheorem{definition}[theorem]{Definition}

\newcommand\remove[1]{}
\newcommand{\nc}{\newcommand}

\allowdisplaybreaks
\nc\bfa{{\boldsymbol a}}\nc\bfA{{\bf A}}\nc\cA{{\mathcal A}}
\nc\bfb{{\boldsymbol b}}\nc\bfB{{\boldsymbol B}}\nc\cB{{\mathcal B}}
\nc\bfc{{\boldsymbol c}}\nc\bfC{{\bf C}}\nc\cC{{\mathcal C}}
\nc\sC{{\mathscr C}}
\nc\bfd{{\boldsymbol d}}\nc\bfD{{\bfD}}
\nc\cD{{\mathcal D}}
\nc\bfe{{\boldsymbol e}}\nc\bfE{{\bf E}}\nc\cE{{\mathcal E}}
\nc\bff{{\boldsymbol f}}\nc\bfF{{\bf F}}\nc\cF{{\mathcal F}}
\nc\bfg{{\boldsymbol g}}\nc\bfG{{\bf G}}\nc\cG{{\mathcal G}}
\nc\bfh{{\boldsymbol h}}\nc\bfH{{\bf H}}\nc\cH{{\mathcal H}}
\nc\bfi{{\boldsymbol i}}\nc\bfI{{\bf I}}\nc\cI{{\mathcal I}}\nc\sI{{\mathscr I}}
\nc\bfj{{\boldsymbolj}}\nc\bfJ{{\bf J}}\nc\cJ{{\mathcal J}}
\nc\bfk{{\boldsymbolk}}\nc\bfK{{\bf K}}\nc\cK{{\mathcal K}}
\nc\bfl{{\boldsymboll}}\nc\bfL{{\bf L}}\nc\cL{{\mathcal L}}
\nc\bfm{{\boldsymbolm}}\nc\bfM{{\bf M}}\nc\cM{{\mathcal M}}
\nc\bfn{{\boldsymboln}}\nc\bfN{{\bf N}}\nc\cN{{\mathcal N}}
\nc\bfo{{\boldsymbolo}}\nc\bfO{{\bf O}}\nc\cO{{\mathcal O}}
\nc\bfp{{\boldsymbolp}}\nc\bfP{{\bf P}}\nc\cP{{\mathcal P}}
\nc\eP{{\EuScriptP}}\nc\fP{{\mathfrak P}}
\nc\bfq{{\boldsymbol q}}\nc\bfQ{{\bf Q}}\nc\cQ{{\mathcal Q}}
\nc\bfr{{\boldsymbol r}}\nc\bfR{{\bf R}}\nc\cR{{\mathcal R}}
\nc\bfs{{\boldsymbol s}}\nc\bfS{{\boldsymbol S}}\nc\cS{{\mathcal S}}
\nc\bft{{\boldsymbol t}}\nc\bfT{{\bf T}}\nc\cT{{\mathcal T}}
\nc\bfu{{\boldsymbol u}}\nc\bfU{{\bf U}}\nc\cU{{\mathcal U}}
\nc\bfv{{\boldsymbol v}}\nc\bfV{{\bf V}}\nc\cV{{\mathcal V}}
\nc\bfw{{\boldsymbol w}}\nc\bfW{{\bf W}}\nc\cW{{\mathcal W}}
\nc\bfx{{\boldsymbol x}}\nc\bfX{{\bf X}}\nc\cX{{\mathcal X}}
\nc\bfy{{\boldsymbol y}}\nc\bfY{{\bf Y}}\nc\cY{{\mathcal Y}}
\nc\bfz{{\boldsymbol z}}\nc\bfZ{{\bf Z}}\nc\cZ{{\mathcal Z}}

\begin{document}

\tikzset{
  block/.style    = {draw, thick, rectangle},
  round/.style      = {draw, circle,node distance = 2cm}, 
}
\newcommand{\suma}{\Large$+$}

\title{\vspace*{0.25in} Universal Source Polarization and an Application to a Multi-User Problem }

\author{\IEEEauthorblockN{Min Ye}
\IEEEauthorblockA{Department of ECE/ISR\\
University of Maryland\\
Email: yeemmi@gmail.com}
\and
\IEEEauthorblockN{Alexander Barg}
\IEEEauthorblockA{Department of ECE/ISR\\
University of Maryland\\
Email: abarg@umd.edu}}

\maketitle
{\renewcommand{\thefootnote}{}\footnotetext{

\vspace{-.2in}
 
\noindent\rule{1.5in}{.4pt}

\noindent$^\ast$ Research supported in part by NSF grants CCF1217245 and CCF1217894.
}}
\renewcommand{\thefootnote}{\arabic{footnote}}
\setcounter{footnote}{0}

\begin{abstract}
We propose a scheme that universally achieves the smallest possible compression rate for a class of sources with side information,
and develop an application of this result for a joint source channel coding problem over a broadcast channel.
\end{abstract}

\maketitle

\section{Introduction}
Polar coding, introduced by Ar{\i}kan in \cite{Arikan09}, has attracted much attention for its ability to achieve capacity of binary-input memoryless output-symmetric channels under a low-complexity decoding procedure (a successive cancellation, or SC decoder). The polar code construction has been extended in a
number of ways including non-binary alphabets as well as asymmetric channels, source coding problems, and various multi-user 
scenarios. The original construction of polar codes depends on the communication channel for the transmission.
At the same time, it is often desirable to have a coding scheme that attains capacity of channels irrespective of the 
structure of the transition probabilities. This feature, termed universal coding, has been recently studied in 
several works. One line of research was started in the works by Korada \cite{Korada09} who proved that polar codes constructed
for a given communication channel $V$ can support reliable transmission with SC decoding over a channel $W$ that is a stochastically degraded
version of $V$.  
Sutter and Renes \cite{Sutter13} extended this result to channels that are ``less noisy'' with respect to the original channel.
Similar results for universal source polarization were obtained earlier by Abbe \cite{Abbe10}.
These works assume that the receiver has full knowledge of the channel/source statistics.
In a recent work Alsan \cite{Alsan13} considered conditions for reliable communication with polar codes when both the encoder and 
decoder are designed for a channel different from the actual communication channel (``mismatched decoding'').

Another line of works is concerned with the classical definition of universal coding, 
aiming to achieve compound capacity of a set of channels using modified polar codes. Along these lines, Hassani and Urbanke showed \cite{Hassani09} that polar codes under SC decoding cannot achieve the compound capacity of a set of binary-input output-symmetric channels. In a later work
\cite{Hassani13} they proposed several modifications of the original polar coding scheme 
that achieve the compound capacity at the cost of increasing the decoding delay. A similar result was also established
by {\c S}a{\c s}o{\u g}lu and Wang \cite{Wang13} (see also Mahdavifar et al.~\cite{mah13}).

In this paper we address the universality problem for source coding with 
side information using polar codes, and use this scheme to attain the region of achievable rates in a joint source-channel coding problem proposed by Tuncel in \cite{Tuncel06}.
The formal definition of the universal source coding problem with side information is given in Section \ref{2}, where we also
present our coding scheme. The scheme itself forms a modification of the ``chaining'' idea from \cite{Hassani13}, adapted
for the source coding problem. The block length optimization problem is analyzed in Sect. \ref{op}.
The second main result of this work relates to a joint source-channel coding problem over a broadcast channel. In Sect. \ref{4}
we design a polar-codes-based compression scheme
to construct rate-optimal codes for this problem. 
As a preliminary result, in Section \ref{3} we analyze a generalized universal compression problem in which the decoders have access
to the encoding sequences produced by some subsets of the set of encoders in the system.


\section{Problem Statement and coding scheme}\label{2}
\subsection{Problem Statement and Preliminaries}\label{s1}
In what follows, the index set $\{1,...,n\}$ is abbreviated as $[n]$. For a subset $\cA \subseteq [n]$ 
we denote by $\cA^c$ its complement in $[n]$ and abbreviate $\{X^i\}_{i \in \cA}$ as $X^{\cA}.$ 
If $\cA=\{i,i+1,\dots,j-1,j\}$, we write $X^{i:j}$ instead of $X^\cA$.

Before we state the main theorem of this paper, we first recall some preliminaries on polar codes.
For random variables $(X,Y) \sim P_{X,Y}$ over finite alphabets $\cX \times \cY$, the average error probability of the maximum a posteriori estimator $\hat{x}(y)= \arg \max_{x \in \cX}P_{X|Y}(x|y)$ of $X$ given $Y$ is
$$
P_e(X|Y)=1-\sum_{y \in \cY}P_Y(y) \max_{x \in \cX} P_{X|Y}(x|y).
$$
The conditional entropy $H_q(X|Y)$ is defined as
$$
H_q(X|Y)=- \sum_{x \in \cX. y \in \cY} P_{X,Y}(x,y) \log_q P_{X|Y}(x|y).
$$
Assume that $N=2^m$ for some integer $m$, and define the polarizing transform $G_N=G^{\otimes m},$ 
where $G=\text{\small{$\Big(\hspace*{-.05in}\begin{array}{c@{\hspace*{0.05in}}c}
    1&0\\1&1\end{array}\hspace*{-.05in}\Big)$}}$ 
and $\otimes$ denotes the Kronecker product of matrices. 
Given the vector $(X^{1:N}, Y^{1:N})$ of $N$ independent copies of the random 
 variables $(X,Y)$, define
a random vector $U^{1:N}=X^{1:N}G_N.$ For $\beta \in (0,1/2)$, consider the set
   \begin{equation}\label{eq:le}
\cL_{X|Y}^{(N)}=\{i \in [N]:P_e(U^i|U^{1:(i-1)},Y^{1:N}) \leq 2^{-N^{\beta}}\}
    \end{equation}
consisting of indices such that $U^i$ is approximately a deterministic function of $(U^{1:(i-1)},Y^{1:N})$. When $\cX=\{0,1\}$, Ar{\i}kan \cite{Arikan10} \cite{ari09a} proved that
$$
\lim_{m \to \infty} \frac{1}{N} | \cL_{X|Y}^{(N)} |=1-H_2(X|Y).
$$
In \cite{Mori14}, Mori and Tanaka proved that
\begin{equation}\label{cardinality}
\lim_{m \to \infty} \frac{1}{N} | \cL_{X|Y}^{(N)} |=1-H_q(X|Y).
\end{equation}
for the case when $\cX={\mathbb F}_q$ is a prime field (the channel coding version of this result first appeared
in \cite{sas09a}). Mori and Tanaka also extended this result to the case when $\cX={\mathbb F}_q$ is a finite field
(not necessarily prime) in the following way. Let $\alpha$ be a primitive element of ${\mathbb F}_q$ 
and let $U^{1:N}=X^{1:N} G^{\otimes m},$ where this time
$G=\text{\small{$\Big(\hspace*{-.05in}\begin{array}{c@{\hspace*{0.05in}}c}
    1&0\\\alpha&1\end{array}\hspace*{-.05in}\Big)$}}.$ Again define the set $\cL_{X|Y}^{(N)}$ of low-entropy symbols $U^i$ using 
\eqref{eq:le}. As shown in \cite{Mori14}, the limit relation \eqref{cardinality} still holds true, providing a basis for a construction
of $q$-ary polar codes.

The universal source coding problem can be formulated as follows. Consider a collection  of random variables $(X, Y_1,...,Y_K)$ with joint distribution $P_{X,Y_1,...,Y_K}$ defined on the set $\cX \times \cY_1 \times \cdots \times \cY_K$, where $\cX={\mathbb F}_q$ is a finite field of size $q$ and $\cY_1,...,\cY_K$ are arbitrary finite sets.
We consider $X$ as a memoryless source and  $Y_1,...,Y_K$ as local side information values
about $X$ available to decoders $1,\dots, K,$ respectively. 
Let $(X^{1:n},Y_1^{1:n},...,Y_K^{1:n})$ be $n$ independent copies of $(X,Y_1,...,Y_K)$. 
The encoder and all the decoders have knowledge of the joint distribution $P_{X,Y_1,...,Y_K}$. 
The realization of side information $Y_k^{1:n}$ is only available to decoder $k, k=1,\dots, K.$ 
The encoder aims at conveying $X^{1:n}$ to all the decoders in a lossless way, i.e., with an arbitrarily small probability of error.
\remove{
 way (i.e., with vanishing error probability of decoding).
We say that the rate $R$ is achievable if their exist a sequence of encoders
  \begin{equation}\label{eq:f}
f^{(n)}:\cX^n \to [q^{Rn}]
  \end{equation}
and $K$ sequences of decoders
   \begin{equation}\label{eq:g}
g_k^{(N)}:[q^{Rn}] \times \cY_k^N \to \cX^N, 1\leq k \leq K
  \end{equation}
such that the probability of error
$$
P_k^{(N)}=\Pr[X^N \neq g_k^{(N)}(f^{(N)}(X^N),Y_k^N)]
$$
vanishes as $N \to \infty$ uniformly for $1 \leq k \leq K,$  while 
   $$
   \lim_{N \to \infty}\frac{\log M_N}{N}=R.
   $$
}

\subsection{Informal description: The case of $K=2$} \label{basic}
Before formalizing the solution to the problem, we give a brief informal description. Let us consider the case $K=2$. For a sequence $x^{1:N}$ of length $N=2^m$, let $u^{1:N}=x^{1:N}G_N.$ Define two mappings $f_{{1}}^N$ and $f_{{2}}^N$ as follows:
  $$
  f_{{1}}^N(x^{1:N})=u^{(\cL_{X|Y_1}^{(N)})^c}, \quad f_{{2}}^N(x^{1:N})=u^{(\cL_{X|Y_2}^{(N)})^c}.
  $$  
Given a pair $(f_{i}^N(X^{1:N}),Y_i^{1:N}), i=1,2$, we can recover $X^{1:N}$ with a small error probability for $N$ sufficiently large. 

Consider a source sequence $x^{1:n}$ of length $n$, where $n=tN$ and $N=2^m.$ We define the encoding mapping $f_{[2]}^n:{\cX}^n \to {\cX}^{nR}$ as 
\begin{align*}
f_{[2]}^n(x^{1:n})=\{f_{{1}}^N(x^{1:N}),\{f_{{1}}^N (x^{(iN+1):((i+1)N)})\oplus\\ f_{{2}}^N(x^{((i-1)N+1):(iN)})
\}_{i=1}^{t-1}, f_{{2}}^N(x^{(n-N+1):n})\}, 
\end{align*}
where $\oplus$ denotes coordinate-wise addition over the finite field $\cX={\mathbb F}_q.$ If the two sequences have different length, 
we first pad the shorter sequence with zeros then perform the addition.

The rate of this coding scheme satisfies
$$
    R\le\frac{t+1}{t} (\max(H_q(X|Y_1),H_q(X|Y_2))+o(1))
$$
as $N \to \infty$. We can approach the rate $\max_{k\in\{1,2\}}H_q(X|Y_k)$ by choosing a 
sufficiently large ``chaining'' parameter $t.$ 

The decoding scheme is as follows (Fig.~\ref{myf}).
The first decoder has knowledge of $f_{{1}}^N(x^{1:N})$ and $y_1^{1:N}$. Thus, it can recover the sequences $u^{1:N}$ and $x^{1:N}$ with a small error probability. 
As a result, it can calculate $f_{{2}}^N(x^{1:N})$. 
Since $f_{[2]}^n(x^{1:n})$ contains $f_{{1}}^N(x^{(N+1):(2N)}) \oplus f_{{2}}^N (x^{1:N}))$, 
the first decoder now has knowledge of $f_{{1}}^N(x^{(N+1):(2N)})$. Together with the side information $y_1^{(N+1):(2N)},$ it can losslessly decode $x^{(N+1):(2N)}$ and calculate the $f_{{2}}^N(x^{(N+1):(2N)})$. 
The first decoder 
then iterates this procedure: After decoding $x^{(iN-N+1):(iN)}$, it calculates $f_{{2}}^N(x^{(iN-N+1):(iN)}).$
Using $f_{{2}}^N(x^{(iN-N+1):(iN)}) \oplus f_{{1}}^N (x^{(iN+1):(iN+N)}),$ it finds $f_{{1}}^N (x^{(iN+1):(iN+N)})$ and then decodes $x^{(iN+1):(iN+N)}$. The whole sequence $x^{1:n}$ can be thus recovered.

\begin{SCfigure*}
\centering

\begin{tikzpicture}[auto]
\draw
	node at (0,0) [block] (f1) {$f_1^N(x[1])$}
      node at (2.8,0)[block] (f2) {$f_1^N(x[2]) \oplus f_2^N(x[1])$}
      node at (0,3) [block](x1){$x[1]$}
      node at (1.2,1.5) [block](i1){$f_2^N(x[1])$}
      node at (2.8,1.5) [round](sum1){\suma}
      node at (4.5,1.5) [block](i2){$f_1^N(x[2])$}
      node at (4.5,3) [block](x2){$x[2]$};

	\draw[->](f1) -- node {$y_1[1]$}(x1);
 	\draw[->](i2) -- node {$y_1[2]$} (x2);
	\draw[->](i1) -- node {} (sum1);
	\draw[->](f2) -- node {} (sum1);
	\draw[->](sum1) -- node {} (i2);
	\draw[->](x1) -| node{} (i1);

      \draw node at (5.9, 2) {\Huge \dots};
      \draw node at (5.9, 0) {\Huge \dots};
	
\draw
	node at (9.3,0) [block] (f3) {$f_1^N(x[t]) \oplus f_2^N(x[t-1])$}
      node at (12.5,0) [block] (f4) {$f_2^N(x[t])$}
      node at (6.5,3) (v1) {}
      node at (7.4,1.5) [block] (i3) {$ f_2^N(x[t-1])$}
      node at (10.8,1.5) [block] (i4) {$f_1^N(x[t])$}
      node at (9.3,1.5) [round](sum2){\suma}
      node at (10.8,3) [block](x3){$x[t]$};

	\draw[->](i4) -- node {$y_1[t]$}(x3);
	\draw[->](i3) -- node {} (sum2);
	\draw[->](f3) -- node {} (sum2);
	\draw[->](sum2) -- node {} (i4);
      \draw[->](v1) -| node{} (i3);

\draw
	node at (12.5,-3) [block] (x4) {$x[t]$}
      node at (7.4,-1.5) [block] (i5) {$ f_2^N(x[t-1])$}
      node at (10.8,-1.5) [block] (i6) {$f_1^N(x[t])$}
      node at (9.3,-1.5) [round](sum3){\suma}
      node at (7.4,-3) [block] (x5) {$ x[t-1]$};

	\draw[->](f4) -- node {$y_2[t]$}(x4);
	\draw[->](i5) -- node {$y_2[t-1]$}(x5);
	\draw[->](i6) -- node {} (sum3);
	\draw[->](f3) -- node {} (sum3);
	\draw[->](sum3) -- node {} (i5);
      \draw[->](x4) -| node{} (i6);

      \draw node at (5.9, -2) {\Huge \dots};

\draw
      node at (1.2,-1.5) [block](i7){$f_2^N(x[1])$}
      node at (2.8,-1.5) [round](sum4){\suma}
      node at (4.5,-1.5) [block](i8){$f_1^N(x[2])$}
      node at (5.3,-3) (v2) {}
      node at (1.2,-3) [block](x6){$x[1]$};

	\draw[->](i7) -- node {$y_2[1]$}(x6);
	\draw[->](i8) -- node {} (sum4);
	\draw[->](f2) -- node {} (sum4);
	\draw[->](sum4) -- node {} (i7);
      \draw[->](v2) -| node{} (i8);

	\draw [color=green,thick](-0.9,-0.5) rectangle (13.4,0.5);
	\draw [color=red,thick](-0.9,0.7) rectangle (11.7,3.4);
	\node at (12.7,2.5) {\textcolor{red}{\textsc{Decoder 1}}};
	\draw [color=blue,thick](0.3,-3.4) rectangle (13.4,-0.7);
	\node at (-0.63,-2.5)  {\textcolor{blue}{\textsc{Decoder 2}}};

\end{tikzpicture}

\caption{The decoding procedure. For notational simplicity, we use $x[i], y_1[i], y_2[i]$ to denote $x^{(iN-N+1):(iN)}$, $y_1^{(iN-N+1):(iN)}$, $y_2^{(iN-N+1):(iN)}$ respectively for all $i=1,2,\dots,t.$ }
\label{myf}

\end{SCfigure*}

The second decoder follows a similar procedure in reverse order. With the knowledge of $f_{{2}}^N(x^{(n-N+1):n})$ 
and $y_2^{(n-N+1):n}$, it can recover $u^{(n-N+1):n}$ and $x^{(n-N+1):n}$. 
Then it calculates $f_{{1}}^N(x^{(n-N+1):n})$. Since $f_{[2]}^n(x^{1:n})$ contains 
$f_{{2}}^N(x^{(n-2N+1):(n-N)}) \oplus f_{{1}}^N(x^{(n-N+1):n})$, the second decoder now has knowledge of $f_{{2}}^N(x^{(n-2N+1):(n-N)})$. 
Together with $y_2^{(n-2N+1):(n-N)}$, it can losslessly decode
$x^{(n-2N+1):(n-N)}$ and recover $f_{{1}}^N(x^{(n-2N+1):(n-N)})$. By repeating this procedure the second 
decoder can also recover the entire source sequence. 

\subsection{Formal description of the coding scheme}
Below we use the following notation. In a universal compression scheme, the code is constructed for a set of decoders, such as 
$[K]=\{1,\dots,K\}.$ All the encoding and decoding maps as well as the values of the rate and the error probability
have the same subscript, such as $[K].$  If the compression scheme only contains a single decoder, such as $J+1$, then we write $\{J+1\}$ as 
the subscript. If we need to refer to the block length explicitly, we use a superscript such as $n.$
Of course, any universal scheme uses a single encoder map for all the $K$ decoders. We refer to decoder $k, k\in[K]$ by
introducing
a second superscript, and so the complete notation for a decoder is of the form $g_{[K]}^{n,k}.$ We use similar notation
for the encoder, the values of the rate and of the error probability of decoding.

  \begin{theorem}\label{mth}
For any $\epsilon >0,\delta >0$, 
there are integers $t$ and $m_0$ such that for any $m \geq m_0$ there exists an encoder
  $$
f_{[K]}^{n}:\cX^n \to \cX^{nR_{{[K]}}}, \quad n=t2^m
  $$
and $K$ decoders
$$
g_{[K]}^{n,k}:\cX^{nR_{{[K]}}} \times \cY_k^n \to \cX^n, 1\leq k \leq K
$$
such that the rate satisfies 
$$
0<R_{[K]}- \max_{1 \leq k \leq K}H_q(X|Y_k)< \delta$$
and the 
probability of error satisfies
$$
P_{{[K]}}^{n,k}\triangleq \Pr[X^{1:n} \neq g_{[K]}^{n,k}(f_{{[K]}}^{n}(X^{1:n}),Y_k^{1:n})] < \epsilon
$$
for all $1 \leq k \leq K$.
  \end{theorem}
  \begin{proof}
We prove the theorem by induction. We begin with a slightly 
stronger claim for $K=1$: For any $\epsilon >0, \delta >0,$ and any integer $t$, 
there exists an integer $m_0$ such that for any $m \geq m_0$, we can find an encoder and a decoder that satisfy the conditions
of the theorem. We first consider the case when $\cX={\mathbb F}_q$ is a prime field. 
As before, let $N=2^m$ and $n=tN.$
Given a realization of the source and side information $(x^{1:n},y_1^{1:n})$, 
define the sequence 
   \begin{equation}\label{eq:diag}
   u^{1:n}=x^{1:n}\text{diag}(G_N,G_N,\dots,G_N)
   \end{equation}
where the diagonal matrix is formed of $t$ identical blocks $G_N.$
In each block of length $N$ there is a set of low-entropy symbols as defined in \eqref{eq:le}. Denote the union of these
sets by $\cL_{X|Y_1}^{(n)}$:  
\begin{align*}
\cL_{X|Y_1}^{(n)}=\big\{i \in [n]:(i-(\lceil{i}/{N} \rceil -1)N) \in \cL_{X|Y_1}^{(N)}\big\}.
\end{align*}
As an immediate consequence of \eqref{cardinality} we have that
\begin{equation}\label{cardinality2}
\lim_{m \to \infty} \frac{1}{n} | \cL_{X|Y_1}^{(n)}|=1-H_q(X|Y_1).
\end{equation}
The encoder $f_{\{1\}}^{n}$ is defined as 
   $$
   f_{\{1\}}^{(n)}: x^{1:n} \mapsto \big\{ u^j : j\in\big(\cL_{X|Y_1}^{(n)}\big)^c \big\}.
   $$
The decoder only needs to determine the values $u_j,j\in{\cL_{X|Y_1}^{(n)}}$ since $x^{1:n}$ and $u^{1:n}$ are in one-to-one correspondence. For all $i \in \cL_{X|Y_1}^{(n)}$, the decoder generates its decision as 
    \begin{multline*}
\hat{u}^i=\arg\max_{u \in \cX} \Pr\big(u|u^{(\lceil i/N \rceil N-N +1):(i-1)},\\ y_1^{(\lceil i/N
 \rceil N-N +1):\lceil i/N \rceil N}\big),
   \end{multline*}
where the probability is computed with respect to the random variable $U_i$ conditional on 
$U^{(\lceil i/N \rceil N-N +1):(i-1)}, Y_1^{(\lceil i/N \rceil N-N +1):(\lceil i/N \rceil N)}.$

Now we invoke the results on the error probability of decoding for a ``single'' polar block \cite{ari09a} and
use the union bound to extend it to $t$ such blocks. We conclude that 
the probability of error $P_{\{1\}}^{n,1}<t2^{-N^{\beta}}, 0<\beta<1/2$ which is less that $\epsilon$ for $m$ sufficiently large. 
Together with \eqref{cardinality2}, this implies our claim our claim for $K=1$ and prime $q.$ 
When $q$ is a prime power, we can follow the above arguments upon replacing $G_N$ in \eqref{eq:diag}
with $\text{\small{$\Big(\hspace*{-.05in}\begin{array}{c@{\hspace*{0.05in}}c}
    1&0\\\alpha&1\end{array}\hspace*{-.05in}\Big)$}}^{\otimes m},$
where $\alpha$ is a primitive element of ${\mathbb F}_q.$

\remove{

$\cX={\mathbb F}_q$ is a non-prime field, we only need to change the definition of $u^{1:n}$ to $u^{iN+1:iN+N}=x^{iN+1:iN+N} \text{\small{$\Big(\hspace*{-.05in}\begin{array}{c@{\hspace*{0.05in}}c}
    1&0\\\alpha&1\end{array}\hspace*{-.05in}\Big)$}}^{\otimes m}$ respectively for $0 \leq i \leq t-1$, where  $\alpha$ is a primitive element of ${\mathbb F}_q.$}

This establishes the induction base. Now suppose that the claim of the theorem holds for $K=J$. By the
induction hypothesis, for any $\epsilon >0$ and $\delta >0$, there are an integer $t_1$ and 
a corresponding $m_1$ such that for any $m \geq m_1$ there is an encoder
   $$
f_{[J]}^{n_1}:\cX^{n_1} \to \cX^{n_1R_{{[J]}}}
   $$
and $J$ decoders
   $$
g_{[J]}^{n_1,j}:\cX^{n_1R_{{[J]}}} \times \cY_j^{n_1} \to \cX^{n_1}, 1\leq j \leq J
   $$
such that the block length $n_1=t_1 2^m$, the compression rate satisfies
\begin{equation}\label{rj}
R_{{[J]}}< \max_{1 \leq j \leq J}H_q(X|Y_j)+\delta /2
\end{equation}
and the probability of error
\begin{equation}\label{ej}
\begin{aligned}
P_{[J]}^{n_1,j}&=\Pr[X^{1:n_1} \neq g_{[J]}^{n_1,j}(f_{{[J]}}^{n_1}(X^{1:n_1}),Y_j^{1:n_1})]\\
 &< \epsilon /t_2
\end{aligned}
\end{equation}
for all $1 \leq j \leq J$, where $t_2=\lceil \frac{2}{\delta} \rceil +1$. Moreover, there is an 
$m_2$ such that for any integer $m \geq m_2$ we can find an encoder
$$
f_{\{J+1\}}^{n_1}:\cX^{n_1} \to \cX^{n_1R_{\{J+1\}}}
$$
and a decoder
$$
g_{\{J+1\}}^{n_1,J+1}:\cX^{n_1R_{\{J+1\}}} \times \cY_{J+1}^{n_1} \to \cX^{n_1}
$$
such that $n_1=t_1 2^m$, the rate satisfies 
\begin{equation}\label{rj1}
R_{\{J+1\}}< H_q(X|Y_{J+1})+\delta /2
\end{equation}
and the probability of decoding error satisfies
\begin{equation}\label{ej1}
\begin{aligned}
P_{\{J+1\}}^{n_1,J+1}&=\Pr[X^{1:n_1} \neq g_{\{J+1\}}^{n_1,J+1}(f_{\{J+1\}}^{n_1}(X^{1:n_1}),Y_{J+1}^{1:n_1})] \\
& < \epsilon /t_2.
\end{aligned}
\end{equation}

Now we prove the claim for $K=J+1$. Let $t=t_1t_2, m_0=\max(m_1,m_2)$. For any integer $m \geq m_0$, let $n_1=t_1 2^m, n=t_1t_2 2^m$. We can find encoders $f_{{[J]}}^{n_1}, f_{\{J+1\}}^{n_1}$ and decoders $g_{{[J]}}^{n_1,j}$, $1 \leq j \leq J$, $g_{\{J+1\}}^{n_1,J+1}$ satisfying \eqref{rj}-\eqref{ej1}. We define the encoder $f_{{[J+1]}}^n:\cX^{n} \to \cX^{nR_{{[J+1]}}}$ by Eq.\eqref{eq:e1} at the top of the next page,
where $\oplus$ denotes coordinate-wise addition over ${\mathbb F}_q.$ 
\begin{figure*}[!t]
\begin{equation}\label{eq:e1}
f_{{[J+1]}}^{n}(x^{1:n})=\{f_{{[J]}}^{n_1}(x^{1:n_1}), \{ f_{{[J]}}^{n_1}(x^{(in_1+1):((i+1)n_1)}) \oplus f_{\{J+1\}}^{n_1}(x^{((i-1)n_1+1):(in_1)}) 
 \}_{i=1}^{t_2-1},  f_{\{J+1\}}^{n_1}(x^{(n-n_1+1):n}) \}
\end{equation}
\begin{gather}
\hat{x}_j^{1:n_1}=g_{{[J]}}^{n_1,j}(f_{{[J]}}^{n_1}(x^{1:n_1}),y_j^{1:n_1}), \quad j=1,\dots,J \label{eq:e1-1}
\end{gather}
\begin{multline}
\hat{x}_j^{(in_1+1):(in_1+n_1)}=g_{{[J]}}^{n_1,j} \Big(f_{{[J]}}^{n_1}(x^{(in_1+1):((i+1)n_1)} \oplus f_{\{J+1\}}^{n_1}(x^{((i-1)n_1+1):(in_1)}) 
  )
 \ominus  f_{\{J+1\}}^{n_1}(\hat{x}_j^{((i-1)n_1+1):(in_1)}),\\
 y_j^{(in_1+1):((i+1)n_1)} \Big), \quad \quad i=1,\dots,t_2-1,  j=1,\dots,J \label{eq:e1-2}
 \end{multline}
\begin{gather}
\hat{x}_{J+1}^{(n-n_1+1):n}=g_{\{J+1\}}^{n_1,J+1} \Big( f_{\{J+1\}}^{n_1}(x^{(n-n_1+1):n})  , y_{J+1}^{(n-n_1+1):n}\Big) \label{eq:e2-1}
\end{gather}
\begin{multline}
\hat{x}_{J+1}^{(in_1-n_1+1):(in_1)}=g_{\{J+1\}}^{n_1,J+1} \Big( f_{\{J+1\}}^{n_1}(x^{((i-1)n_1+1):(in_1)})  
 \oplus  f_{{[J]}}^{n_1}(x^{(in_1+1):((i+1)n_1)}) \ominus  f_{{[J]}}^{n_1}(\hat{x}_{J+1}^{(in_1+1):((i+1)n_1)}),\\
 y_{J+1}^{((i-1)n_1+1):(in_1)}\Big), \quad \quad i=t_2-1,\dots,1 \label{eq:e2-2}
\end{multline}
\vspace*{2pt}\hrulefill
\end{figure*}
 Using \eqref{rj}, it is easy to see  that
\begin{align*}
  R_{{[J+1]}} &\leq \frac{(t_2+1)}{t_2} \max(R_{{[J]}},R_{\{J+1\}})\\
&\leq \max(R_{{[J]}},R_{\{J+1\}})+\frac{1}{t_2} \\
&< \max_{1 \leq j \leq J+1}H_q(X|Y_j)+\delta.
\end{align*}
Thus the rate constraint is satisfied. Note that $x^{1:n}$ consists of $t_2$ blocks of length $n_1$. The first $J$ decoders decode successively from block $1$ to block $t_2$ while decoder $J+1$ decodes in reverse order. We define $\hat{x}_j^{1:n}$ as a function of $x^{1:n}$ and $y_j^{1:n}$ successively as shown in Eqns.~\eqref{eq:e1-1},~\eqref{eq:e1-2} at the top of the
next page.
\remove{
$$
\hat{x}_j^{1:n_1}=g_{{[J]}}^{n_1,j}(f_{{[J]}}^{n_1}(x^{1:n_1}),y_j^{1:n_1}) 
$$
   \begin{multline*}
\hat{x}_j^{(in_1+1):(in_1+n_1)}\\=g_{{[J]}}^{n_1,j} \Big( f_{\{J+1\}}^{n_1}(x^{(in_1-n_1+1):(in_1)}) 
  \oplus f_{{[J]}}^{n_1}(x^{(in_1+1):(in_1+n_1)})\\
 \ominus  f_{\{J+1\}}^{n_1}(\hat{x}_j^{(in_1-n_1+1):(in_1)}),
 y_j^{(in_1+1):(in_1+n_1)} \Big)
   \end{multline*}
for $i=1,\dots,t_2-1, 1 \leq j \leq J$, where $\ominus$ denotes coordinate-wise subtraction over the finite field $\cX={\mathbb F}_q$. 
}
Decoders $g_{{[J+1]}}^{n,j}$, $1 \leq j \leq J$ are defined as
$$
g_{{[J+1]}}^{n,j}(f_{{[J+1]}}^{n}(x^{1:n}),y_j^{1:n})=\hat{x}_j^{1:n}.
$$
Let $\hat{X}_j^{1:n}=g_{{[J+1]}}^{n,j}(f_{{[J+1]}}^{n}(X^{1:n}),Y_j^{1:n})$. The error probability
\begin{align*}
P_{{[J+1]}}^{n,j}&=\Pr[X^{1:n} \neq \hat{X}_j^{1:n}]
=\Pr[X^{1:n_1} \neq \hat{X}_j^{1:n_1}]\\ 
&\hspace*{-.4in}
\text{\small $+\sum_{i=1}^{t_2-1} \Pr[X^{(in_1+1):(in_1+n_1)} \neq \hat{X}_j^{(in_1+1):(in_1+n_1)},X^{1:in_1}=\hat{X}_j^{1:in_1}]$} \\
&\leq \Pr[X^{1:n_1} \neq \hat{X}_j^{1:n_1}]\\
&\hspace*{-.4in}
\text{\small $+\sum_{i=1}^{t_2-1} \Pr[X^{(in_1+1):(in_1+n_1)} \neq \hat{X}_j^{(in_1+1):(in_1+n_1)}|X^{1:in_1}=\hat{X}_j^{1:in_1}].$}
\end{align*}
It is easy to see that each term in the right hand side of the inequality is equal to $P_{{[J]}}^{n_1,j}$. By \eqref{ej} we conclude that $P_{{[J+1]}}^{n,j} < \epsilon$ for $1 \leq j \leq J$. As for decoder $J+1$, define $\hat{x}_{J+1}^{1:n}$ as a function of $x^{1:n}$ and $y_{J+1}^{1:n}$ successively as shown in Eqns.~\eqref{eq:e2-1},~\eqref{eq:e2-2} at the top of the next page.
\remove{
$$
\hat{x}_{J+1}^{(n-n_1+1):n}=g_{\{J+1\}}^{n_1,J+1}( f_{\{J+1\}}^{n_1}(x^{(n-n_1+1):n}),y_{J+1}^{(n-n_1+1):n}) 
$$
\begin{align*}
&\hat{x}_{J+1}^{(in_1-n_1+1):(in_1)}=g_{\{J+1\}}^{n_1,J+1} \left( f_{\{J+1\}}^{n_1}(x^{(in_1-n_1+1):(in_1)})  \right. \\
& \oplus  f_{{[J]}}^{n_1}(x^{(in_1+1):(in_1+n_1)}) \ominus  f_{{[J]}}^{n_1}(\hat{x}_{J+1}^{(in_1+1):(in_1+n_1)}),\\
&\left. y_{J+1}^{(in_1-n_1+1):(in_1)}\right)
\end{align*}}
Decoder $g_{{[J+1]}}^{n,J+1}$ is defined as
$$
g_{{[J+1]}}^{n,J+1}(f_{{[J+1]}}^{n}(x^{1:n}),y_{J+1}^{1:n})=\hat{x}_{J+1}^{1:n}.
$$
The error probability $P_{{[J+1]}}^{n,J+1}$ can be bounded in exactly the same way as above. Thus we conclude that $P_{{[J+1]}}^{n,j} < \epsilon$ for $1 \leq j \leq J+1$ and complete the proof.
\end{proof}

Note that our proof is constructive. The coding scheme proposed above inherits the low encoding and decoding complexity of polar codes and achieves the rate $\max_{1 \leq k \leq K} H_q(X|Y_k)$.
By the Slepian-Wolf theorem \cite{Slepian73} this is the smallest achievable
rate, so the scheme achieves the optimal encoding rate of universal source coding with side information. Furthermore, in this proof polar codes are used only for establishing the induction base. They can be replaced by any source code whose rate achieves the conditional entropy function.


\subsection{Universal Compression without Side Information}
The compression scheme described above can be carried over without changes 
to solve the classical universal source coding problem. 
Let $X$ be a source over the alphabet $\cX={\mathbb F}_q$. Consider a finite set of distributions 
$\{P_1, P_2,..., P_K\}$ on $\cX$. Suppose that the encoder only knows that the distribution of 
$X$ belongs to this set, while the decoder has knowledge of the actual source distribution. We 
are seeking a lossless compression scheme achieving the ``compound rate'' 
$\max_{1 \leq k \leq K}H_q(P_k)$, where $H_q(P_k)$ denotes the entropy of distribution $P_k$. 
Without loss of generality we can assume that $H_q(P_1)=\max_{1 \leq k \leq K}H_q(P_k)$. 
When $\cX=\{0,1\}$, Abbe \cite{Abbe10} proved that a polar code constructed for the distribution 
$P_1$ will compress losslessly for all $P_k, 1 \leq k \leq K$. 
This fact relies on special properties of the binary source alphabet.
It is further shown in \cite{Abbe10} that polar codes do not achieve
the compound rate for source coding when the size of the source alphabet is larger than $2$.
In contrast, our scheme, although requiring a much larger block length, is able to achieve the compound rate for any finite set of distributions whenever $\cX$ is a finite field. Formally, we have the following proposition, whose proof follows the same steps as the proof of Theorem \ref{mth} and is therefore omitted here.

\begin{proposition}
For any $\epsilon >0$ and $\delta >0$, there are integers $t$ and $m_0$ such that for any $m \geq m_0,$ we can find a length $n=t 2^m$ block encoder
$$
f_{{[K]}}^{n}:\cX^n \to \cX^{nR_{{[K]}}}
$$
and $K$ decoders
$$
g_{{[K]}}^{n,k}:\cX^{nR_{{[K]}}} \to \cX^n, 1\leq k \leq K
$$
such that the rate satisfies $R_{{[K]}}< \max_{1 \leq k \leq K}H_q(P_k)+\delta$ and the probability of error satisfies
  $$
P_{{[K]}}^{n,k}\triangleq\Pr[X^{1:n} \neq g_{{[K]}}^{n,k}(f_{{[K]}}^{n}(X^{1:n}))] < \epsilon
  $$
for all $1 \leq k \leq K$.
\end{proposition}

\section{Length Optimization} \label{op}
We note that the chaining can be done in many ways, leading to codes of different
length. Here we point out how this property can be used to make the block 
length smaller.

Let us consider a special case where $H(X|Y_k)=H_0$ for all $1 \leq k \leq K$. From the analysis above, when we add decoder $i+1$, we need to chain $t_i$ blocks. Doing so 
increases the rate by approximately $H_0/t_i$. Suppose that the rate of each 
underlying polar encoding is $R_1=\max_{1\le i\le K}\frac{|{(\cL_{X|Y_i}^{(N)})}^c|}{N},$ 
while the target rate is $R_2>R_1.$ Letting $\Delta=R_2-R_1$, we must ensure that
\begin{equation}\label{rate}
\sum_{i=1}^{K-1} \frac{H_0}{t_i} \leq \Delta.
\end{equation}
The total block length is $Nt_1 \cdot\cdot\cdot t_{K-1}$, where $N$ is the length of the underlying polar block. By the arithmetic mean-geometric mean inequality we
obtain
\begin{equation}\label{length}
N\prod_{i=1}^{K-1}t_i \geq N\left( \frac{(K-1)H_0}{\Delta} \right)^{K-1}.
\end{equation}
The total block length is minimized by setting $t_i=(K-1)H_0 / \Delta$ for all $1 \leq i \leq K-1$.

Now let us consider different ways to chain the code blocks. 
For simplicity suppose that the number of decoders $K$ is a power of $2$. 
Recall that we first perform the chaining for decoders 1 and 2, then for
decoders 1,2 and 3, etc. This procedure can be naturally represented by a
tree as illustrated in the left part of Fig. \ref{f1}. 
We need to do the chaining construction whenever we increase the depth of the encoding tree.
Analogously to \eqref{rate}, for general encoding tree we have $\sum_{i=1}^{D} \frac{H_0}{t_i} \leq \Delta,$
where $D$ is the depth of the tree and $t_i$ is the number of chaining blocks at nodes of depth $i$.  Proceeding as in \eqref{length} we find $N\prod_{i=1}^{D}t_i \geq N\left( \frac{DH_0}{\Delta} \right)^D.$
It is clear that in order to obtain shorter encoding blocks we must make $D$
as small as possible. This is accomplished by balancing the encoding tree through
rearranging the chaining steps. This rearranging results into the tree depth
$D=\lceil\log_2K\rceil.$  Performing the chaining as shown in the right part of Fig.~\ref{f1} will reduce the total block length from $N\big( \frac{(K-1)H_0}{\Delta} \big)^{K-1}$ to 
   $$
N\Big( \frac{\lceil\log_2K\rceil H_0}{\Delta} \Big)^{\lceil\log_2K\rceil}.
   $$

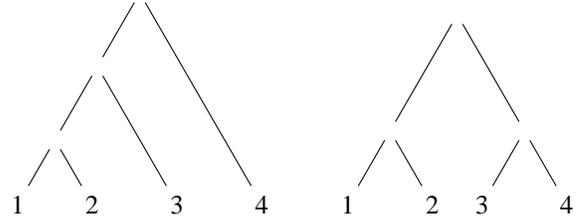
\begin{figure}
\centering
\begin{forest}
  [,calign=fixed edge angles,
calign primary angle=-30, calign secondary angle=30
  [ ,calign=fixed edge angles,
calign primary angle=-30, calign secondary angle=30
[,calign=fixed edge angles,
calign primary angle=-30, calign secondary angle=30[1,tier=b][2,tier=b]][3,tier=b]][4,tier=b]
  ]
\end{forest}\quad\begin{forest} for tree={anchor=center}
  [,calign=fixed edge angles,
calign primary angle=-30, calign secondary angle=30
    [,calign=fixed edge angles,
calign primary angle=-30, calign secondary angle=30[1][2]]
 [,calign=fixed edge angles,
calign primary angle=-30, calign secondary angle=30[3][4]]
  ]
\end{forest}
\caption{The original (left) and improved (right) scheme for K=4}
\label{f1}
\end{figure}

\section{A generalization of the universal source coding problem}\label{3}
In this section we generalize the coding scheme proposed in Section \ref{2} to a more complicated scenario with the aim of using the results
developed here in a joint source-channel coding problem addressed in the next section.

The notation $(X^{1:n},{Y_1}^{1:n},...,{Y_K}^{1:n})$ below has the same meaning as in Section \ref{s1}. We consider the following 
source coding problem. There are $2^{K_1}-1$ encoders and $K$ decoders, where $K_1 \ge K$. The joint distribution $P_{X,Y_1,\dots,Y_K}$ is known to all the encoders and decoders. The encoders $f_{{[K]}}^{n,\cA}:\cX^n \to \cX^{nR_{{[K]}}^{\cA}}$ are indexed by subsets $\cA \subset [K_1],\cA \neq \emptyset$, where $R_{{[K]}}^{\cA}$ is the rate of the encoder. Define 
  $$
  R_{{[K]}}^k=\sum\limits_{\cA: k\in \cA }R_{{[K]}}^{\cA}.
  $$
  The $k$th decoder only has access to $f_{{[K]}}^{n,\cA}(X^{1:n})$ for $\cA \ni k$ and the realization of $Y_k^{1:n}$. Note that none of the $K$ decoders uses the encoded sequence $f_{{[K]}}^{n,\cA}(X^{1:n})$ for $\cA \subset [K_1] \backslash [K].$ We define these encoders simply for the convenience of the following inductive proof. We now present a coding scheme such that all the decoders can recover $X^{1:n}$ losslessly.

For simplicity, we assume that the source alphabet $\cX=\{0,1\},$ and only use the binary entropy function hereafter,
omitting the subscript $q$ from the notation.

\begin{definition}\label{def:cover}
Let $K_1\geq K$ and let $\{a_k\}_{k=1}^{K}$, $\{\bar{R}^{\cA}\}_{\cA \subset [K_1],\cA \neq \emptyset}$ be two sets of nonnegative real numbers. We say that $\{\bar{R}^{\cA}\}_{\cA \subset [K_1],\cA \neq \emptyset}$ \emph{covers} $\{a_k\}_{k=1}^K$ if
  $$
\sum_{\cA \ni k}\bar{R}^{\cA}>a_k
  $$
for $1 \leq k \leq K.$
\end{definition}

\begin{proposition}\label{prop}
Let $K_1\geq K$ and let $\epsilon>0$. Suppose that
a set of nonnegative real numbers $\{\bar{R}^{\cA}\}_{\cA \subset [K_1],\cA \neq \emptyset}$ 
covers $\{H(X|Y_k)\}_{k=1}^K.$ There are integers $t$ and $m_0$ such that for any $m \geq m_0,$ there exist encoders
  $$
f_{ [K]}^{n,\cA}:\cX^n \to \cX^{nR_{{[K]}}^{\cA}}, \quad\cA \subset [K_1],\cA \neq \emptyset
  $$
and $K$ decoders
  $$
g_{{[K]}}^{n,k}:\cX^{nR_{{[K]}}^k} \times \cY_k^n \to \cX^n, 1\leq k \leq K
  $$
such that $n=t2^m,$ the rates satisfy $R_{{[K]}}^{\cA}< \bar{R}^{\cA}$ for all $\cA \subset [K_1],\cA \neq \emptyset,$ 
and the probability of error satisfies
\begin{align*}
P_{{[K]}}^{n,k}=\Pr[X^{1:n} \neq g_{{[K]}}^{n,k} ( \{f_{{[K]}}^{n,\cA}(X^{1:n})\}_{\cA \ni k},Y_k^{1:n})] < \epsilon
\end{align*}
for all $1 \leq k \leq K$.
\end{proposition}

\begin{proof}
We prove the proposition by induction, following the ideas in Theorem \ref{mth}. Let us first establish a stronger claim for $K=1$, namely, for any $\epsilon >0$ and any integer $t$, there is an integer $m_0$ such that for any $m \geq m_0$, we can find encoders and a decoder satisfying the above conditions. No separate
proof is needed because the rate and error probability constraints are essentially the same as those in Theorem \ref{mth} 
in the case of $K=1.$

Let us make the induction step. Suppose that the claim holds for $K=J.$ 
Without loss of generality we assume that $\bar{R}^{\cA}>0$ for all $\cA \subset [K_1],\cA \neq \emptyset$. 
 If $\{\bar{R}^{\cA}\}_{\cA \subset [K_1],\cA \neq \emptyset}$ covers $\{H(X|Y_j)\}_{j=1}^{J+1}$, then there exists another set of positive numbers $\{\hat{R}^{\cA}\}_{\cA \subset [K_1],\cA \neq \emptyset}$ covering $\{H(X|Y_j)\}_{j=1}^{J+1}$ and satisfying $\hat{R}^{\cA}<\bar{R}^{\cA}$ for all $\cA \subset [K_1],\cA \neq \emptyset$. Since $\{\hat{R}^{\cA}\}_{\cA \subset [K_1],\cA \neq \emptyset}$ covers $\{H(X|Y_j)\}_{j=1}^{J+1}$, it also covers $\{H(X|Y_j)\}_{j=1}^J$. By the induction hypothesis,
 for any $\epsilon >0$, there are an integer $t_1$ and a corresponding $m_1$ such that for any $m \geq m_1$ we can find encoders
   $$
f_{{[J]}}^{n_1,\cA}:\cX^{n_1} \to \cX^{n_1R_{{[J]}}^{\cA}}, \quad \cA \subset [K_1],\cA \neq \emptyset
   $$
and $J$ decoders
   $$
g_{{[J]}}^{n_1,j}:\cX^{n_1R_{{[J]}}^j} \times \cY_j^{n_1} \to \cX^{n_1}, 1\leq j \leq J
   $$
where the block length $n_1=t_1 2^m$, the rate satisfies
\begin{equation}\label{srj}
R_{{[J]}}^{\cA}< \hat{R}^{\cA}
\end{equation}
for all $\cA \subset [K_1],\cA \neq \emptyset,$ and the probability of error satisfies
\begin{equation}\label{sej}
\begin{aligned}
P_{{[J]}}^{n_1,j} &=\Pr[X^{1:n_1} \neq g_{{[J]}}^{n_1,j} ( \{ f_{{[J]}}^{n_1,\cA} (X^{1:n_1})\}_{\cA \ni j},Y_j^{1:n_1})]\\
& < \epsilon /t_2
\end{aligned}
\end{equation}
for all $1 \leq j \leq J$, where 
\begin{equation}\label{st2}
t_2=\Big\lceil \frac{1}{\min_{\cA \subset [K_1]}(\bar{R}^{\cA}-\hat{R}^{\cA})} \Big\rceil +1.
\end{equation}
Moreover, there is an $m_2$ such that for any integer $m \geq m_2$ we can find encoders
$$
f_{ \{J+1\}}^{n_1,\cA}:\cX^{n_1} \to \cX^{n_1R_{\{J+1\}}^{\cA}}, \quad \cA \subset [K_1],\cA \neq \emptyset
$$
and a decoder
$$
g_{\{J+1\}}^{n_1,J+1}:\cX^{n_1R_{\{J+1\}}^{J+1}} \times \cY_{J+1}^{n_1} \to \cX^{n_1}
$$
in which $R_{\{J+1\}}^{J+1}=\sum_{\cA \ni J+1}R_{\{J+1\}}^{\cA}$, the block length $n_1=t_1 2^m$, the rate
\begin{equation}\label{srj1}
R_{\{J+1\}}^{\cA}< \hat{R}^{\cA}
\end{equation}
and the probability of error
\begin{equation}\label{sej1}
\begin{aligned}
&P_{\{J+1\}}^{n_1,J+1} \\
&=\Pr[X^{1:n_1} \neq g_{\{J+1\}}^{n_1,J+1} ( \{f_{ \{J+1\}}^{n_1,\cA}(X^{1:n_1})\}_{\cA \ni J+1},Y_{J+1}^{1:n_1})]\\
& < \epsilon /t_2.
\end{aligned}
\end{equation}

Now let us prove the claim for $K=J+1$. Let $t=t_1t_2, m_0=\max(m_1,m_2)$. For any integer $m \geq m_0$, let $n_1=t_1 2^m$ 
and $n=t_1t_2 2^m.$ 
We can find encoders $f_{{[J]}}^{n_1,\cA}, f_{ \{J+1\}}^{n_1,\cA},\cA \subset [K_1],\cA \neq \emptyset$ and 
decoders $g_{{[J]}}^{n_1,j}$, $1 \leq j \leq J$, $g_{\{J+1\}}^{n_1,J+1}$ satisfying \eqref{srj}-\eqref{sej} 
and \eqref{srj1}-\eqref{sej1}. Define encoders $f_{{[J+1]}}^{n,\cA}:\cX^{n} \to \cX^{nR_{{[J+1]}}^{\cA}}$ 
for all $\cA \subset [K_1],\cA \neq \emptyset$ as shown in \eqref{eq:e3} at the top of the next page.
\begin{figure*}[!t]
\begin{equation}\label{eq:e3}
f_{{[J+1]}}^{n,\cA}(x^{1:n})=\{f_{{[J]}}^{n_1,\cA}(x^{1:n_1}), \{f_{{[J]}}^{n_1,\cA}(x^{(in_1+1):((i+1)n_1)}) 
\oplus f_{ \{J+1\}}^{n_1,\cA}(x^{((i-1)n_1+1):(in_1)})  \}_{i=1}^{t_2-1}, f_{ \{J+1\}}^{n_1,\cA}(x^{(n-n_1+1):n}) \}
\end{equation}
\begin{gather}\label{eq:e4-1}
\hat{x}_j^{1:n_1}=g_{{[J]}}^{n_1,j}(\{f_{{[J]}}^{n_1,\cA}(x^{1:n_1})\}_{\cA \ni j},y_j^{1:n_1}), \quad j=1,\dots,J
\end{gather}
    \begin{multline}
\hat{x}_j^{(in_1+1):((i+1)n_1)}
=g_{{[J]}}^{n_1,j} \Big( \Big\{f_{\{J+1\}}^{n_1,\cA}(x^{((i-1)n_1+1):(in_1)}) \oplus f_{{[J]}}^{n_1,\cA}(x^{(in_1+1):((i+1)n_1)}) \ominus  f_{ \{J+1\}}^{n_1,\cA}(\hat{x}_j^{((i-1)n_1+1):(in_1)})\Big\}_{\cA \ni j},\\ y_j^{(in_1+1):((i+1)n_1)}\Big), \quad \quad i=1,\dots,t_2-1, \quad j=1,\dots,J
 \label{eq:e4-2}
   \end{multline}
   \begin{equation}\label{eq:e5}
\hat{x}_{J+1}^{(n-n_1+1):n}
=g_{\{J+1\}}^{n_1,J+1}( \{ f_{ \{J+1\}}^{n_1,\cA} (x^{(n-n_1+1):n}) \}_{\cA \ni J+1} ,y_{J+1}^{(n-n_1+1):n})
    \end{equation}
   \begin{multline}\label{eq:e6}
\hat{x}_{J+1}^{((i-1)n_1+1):(in_1)}=g_{\{J+1\}}^{n_1,J+1}\Big( \Big\{ f_{ \{J+1\}}^{n_1,\cA}(x^{((i-1)n_1+1):(in_1)})\oplus 
f_{{[J]}}^{n_1,\cA}(x^{(in_1+1):((i+1)n_1)})
 \ominus  f_{{[J]}}^{n_1,\cA}(\hat{x}_{J+1}^{(in_1+1):((i+1)n_1)}) \Big\}_{\cA \ni J+1},\\y_{J+1}^{((i-1)n_1+1):(in_1)}\Big), \quad \quad i=t_2-1,\dots,1
   \end{multline}

%

\vspace*{4pt}\hrulefill
\end{figure*}
\remove{  \begin{multline*}
f_{{[J+1]}}^{n,\cA}(x^{1:n})=\{f_{{[J]}}^{n_1,\cA}(x^{1:n_1}), \{ f_{ \{J+1\}}^{n_1,\cA}(x^{(in_1-n_1+1):(in_1)}) \\
\hspace*{.1in}\oplus f_{{[J]}}^{n_1,\cA}(x^{(in_1+1):(in_1+n_1)}) \}_{i=1}^{t_2-1}, f_{ \{J+1\}}^{n_1,\cA}(x^{(n-n_1+1):n}) \},
   \end{multline*}
where $\oplus$ denotes coordinate-wise addition. For two sequences with different lengths, we first add zeros at the end of the shorter sequence to equalize them then perform the coordinate-wise addition.}
 By \eqref{srj}, \eqref{st2}, and \eqref{srj1} we have
   \begin{align*}
R_{{[J+1]}}^{\cA} &\leq \frac{(t_2+1)n_1 \max(R_{{[J]}}^{\cA},R_{\{J+1\}}^{\cA})}{t_2n_1}\\
&\leq \max(R_{{[J]}}^{\cA},R_{\{J+1\}}^{\cA})+\frac{1}{t_2}\\
& < \bar{R}^{\cA}
  \end{align*}
for all $\cA \subset [K_1],\cA \neq \emptyset$. Thus the rate constraint is satisfied. Note that $x^{1:n}$ consists of $t_2$ blocks of length $n_1$. The first $J$ decoders decode successively from block $1$ to block $t_2$ while decoder $J+1$ decodes in reverse order. We define $\hat{x}_j^{1:n}$ as functions of $x^{1:n}$ and $y_j^{1:n}$ successively as shown in
Eqns.~\eqref{eq:e4-1},~\eqref{eq:e4-2} at the top of the next page.
\remove{
as follows,
$$
\hat{x}_j^{1:n_1}=g_{{[J]}}^{n_1,j}(\{f_{{[J]}}^{n_1,\cA}(x^{1:n_1})\}_{\cA \ni j},y_j^{1:n_1}) 
$$
\begin{align*}
&\hat{x}_j^{(in_1+1):(in_1+n_1)}=g_{{[J]}}^{n_1,j} \left( \{f_{ \{J+1\}}^{n_1,\cA}(x^{(in_1-n_1+1):(in_1)}) \right. \\
&  \oplus f_{{[J]}}^{n_1,\cA}(x^{(in_1+1):(in_1+n_1)}) \ominus  f_{ \{J+1\}}^{n_1,\cA}(\hat{x}_j^{(in_1-n_1+1):(in_1)})\}_{\cA \ni j},\\
&\left. y_j^{(in_1+1):(in_1+n_1)}\right)
\end{align*}
for $i=1,\dots,t_2-1, 1 \leq j \leq J$, where $\ominus$ denotes coordinate-wise subtraction.}
 
Decoders $g_{{[J+1]}}^{n,j}$, $1 \leq j \leq J$ are defined as follows:
   $$
g_{{[J+1]}}^{n,j}( \{ f_{{[J+1]}}^{n,\cA} (x^{1:n}) \}_{\cA \ni j} ,y_j^{1:n})=\hat{x}_j^{1:n}.
   $$
Let $\hat{X}_j^{1:n}=g_{{[J+1]}}^{n,j}( \{ f_{{[J+1]}}^{n,\cA}(X^{1:n}) \}_{\cA \ni j} ,Y_j^{1:n})$. The error probability
    \begin{align*}
&P_{{[J+1]}}^{n,j}=\Pr[X^{1:n} \neq \hat{X}_j^{1:n}]=\Pr[X^{1:n_1} \neq \hat{X}_j^{1:n_1}]\\ 
&+\sum_{i=1}^{t_2-1} \Pr[X^{(in_1+1):(in_1+n_1)} \neq \hat{X}_j^{(in_1+1):(in_1+n_1)},X^{1:in_1}=\hat{X}_j^{1:in_1}] \\
&\leq \Pr[X^{1:n_1} \neq \hat{X}_j^{1:n_1}]\\
&\hspace*{.5in}+\sum_{i=1}^{t_2-1} \Pr[X^{(in_1+1):(in_1+n_1)} \neq \hat{X}_j^{(in_1+1):(in_1+n_1)}\mid \\&\hspace*{2in}X^{1:in_1}=\hat{X}_j^{1:in_1}].
\end{align*}
It is easy to see that each term on the right-hand side of this inequality is equal to $P_{{[J]}}^{n_1,j}$.
On account of \eqref{sej} we conclude that $P_{{[J+1]}}^{n,j} < \epsilon$ 
for $1 \leq j \leq J$. 
As for decoder $J+1$, define $\hat{x}_{J+1}^{1:n}$ as a function of $x^{1:n}$ and $y_{J+1}^{1:n}$ successively
as shown in Eqns.~\eqref{eq:e5},~\eqref{eq:e6} at the top of the next page.
\remove{
\begin{align*}
&\hat{x}_{J+1}^{(n-n_1+1):n}\\
&=g_{\{J+1\}}^{n_1,J+1}( \{ f_{ \{J+1\}}^{n_1,\cA} (x^{(n-n_1+1):n}) \}_{\cA \ni J+1} ,y_{J+1}^{(n-n_1+1):n}) 
\end{align*}
\begin{align*}
&\hat{x}_{J+1}^{(in_1-n_1+1):(in_1)}=g_{\{J+1\}}^{n_1,J+1}( \{ f_{ \{J+1\}}^{n_1,\cA}(x^{(in_1-n_1+1):(in_1)})\oplus \\
&f_{{[J]}}^{n_1,\cA}(x^{(in_1+1):(in_1+n_1)})
 \ominus  f_{{[J]}}^{n_1,\cA}(\hat{x}_{J+1}^{(in_1+1):(in_1+n_1)}) \}_{\cA \ni J+1}, \\
&y_{J+1}^{(in_1-n_1+1):(in_1)})
\end{align*}
for $i=t_2-1,\dots,1$.}
 Decoder $g_{{[J+1]}}^{n,J+1}$ is defined as
$$
g_{{[J+1]}}^{n,J+1}(\{ f_{{[J+1]}}^{n,\cA}(x^{1:n}) \}_{\cA \ni J+1} ,y_{J+1}^{1:n})=\hat{x}_{J+1}^{1:n}.
$$
The error probability $P_{{[J+1]}}^{n,J+1}$ can be bounded in exactly the same way as above. We conclude that $P_{{[J+1]}}^{n,j} < \epsilon$ for $1 \leq j \leq J+1,$ which completes the proof.
\end{proof}
Note that this proof is also constructive. It can be easily seen from the proof that given an achievable rate constraint and an error probability threshold $\epsilon>0,$ we can choose arbitrarily large $t$ and $m$ such that there are encoders and decoders with block length $t2^m$ satisfying these constraints. 

\section{Slepian-Wolf coding over broadcast channels}\label{4}
\subsection{Problem Statement}
We consider the following 
communication problem formulated by Tuncel in \cite{Tuncel06}. Below the notation $(X^{1:n},{Y_1}^{1:n},...,{Y_K}^{1:n})$ has 
the same meaning as in Section \ref{s1}.  
Now the encoder is required to map $X^{1:n}$ to a sequence $U^{1:l}$, where $U$ takes values in $\cU=\{0,1\}$. 
The encoded sequence $U^{1:l}$ is transmitted through a memoryless broadcast channel $W(v_1,...,v_K|u)$ with the input
alphabet $\cU$ and finite output alphabets $\cV_1,\dots,\cV_K.$ Let $V_k^{1:l}$ denote the version of $U^{1:l}$ 
received from the channel by Decoder $k, 1 \leq k \leq K$. 

The decoder uses $V_k^{1:l}$ and $Y_k^{1:n}$ to reconstruct $X^{1:n}.$ We say that rate $\kappa$ (measured in channel uses per symbol) is achievable if there exist a sequence of encoders
$$
f_{{[K]}}^{(l,n)}:\cX^n \to \cU^l
$$
and $K$ sequences of decoders
$$
g_{{[K]}}^{(l,n),k}:\cV_k^l \times \cY_k^n \to \cX^n, 1\leq k \leq K
$$
such that the probability of error
$$
P_{{[K]}}^{(l,n),k}=\Pr[X^{1:n} \neq g_{{[K]}}^{(l,n),k}(V_k^{1:l},Y_k^{1:n})]
$$
vanishes uniformly for $1 \leq k \leq K$ as $n,l \to \infty$ while $\frac{l}{n} \to \kappa$.

In \cite{Tuncel06}, Tuncel proved the following theorem which characterizes the set of achievable rates.
\begin{theorem} \label{thm:T} The value
$\kappa$ is achievable if and only if there exists $P_U(u)$ such that
   \begin{equation}\label{condition}
H(X|Y_k)< \kappa I(U;V_k)
   \end{equation}
for all $1 \leq k \leq K$.
\end{theorem}
In the next section we give an explicit scheme that uses the construction of Section \ref{3} to achieve the coding rates guaranteed
by this theorem.

\subsection{Coding Scheme}
For random variables $(X,Y) \sim P_{X,Y}$, 
where $X$ is binary and $Y$ takes values in arbitrary discrete alphabet $\cY$, define the Bhattacharyya parameter as follows:
$$
Z(X|Y)=2\sum_{y \in \cY}P_Y(y)\sqrt{P_{X|Y}(0|y)P_{X|Y}(1|y)}.
$$

Let $\kappa$ and $P_U(u)$ satisfy condition  \eqref{condition} of Theorem \ref{thm:T}. Assume that $N=2^m$ for some integer $m$. 
Given the random vector $(U^{1:N}, V_1^{1:N},\dots, V_K^{1:N})$ of $N$ independent drawings from the distribution 
$P_{U,V_1,...,V_K}=W(V_1,...,V_K|U)P_U$, define $D^{1:N}=U^{1:N}G_{N}$. For $\beta\in(0,1/2)$, consider the sets
$$
\cL_{U|V_k}^{(N)}=\{i \in [N]:Z(D^i|D^{1:i-1},V_k^{1:N}) \leq 2^{-N^{\beta}}\}
$$
$$
\cH_U^{(N)}=\{i \in [N]:Z(D^i|D^{1:i-1}) \geq 1-2^{-N^{\beta}}\}
$$
$$
\cL_U^{(N)}=\{i \in [N]:Z(D^i|D^{1:i-1}) \leq 2^{-N^{\beta}}\}
$$
for $1 \leq k \leq K$. Let $\cI_k^{(N)}=\cL_{U|V_k}^{(N)} \cap \cH_U^{(N)}$ for $1 \le k \le K$. Owing to the results of
Honda and Yamamoto \cite{Honda13}, we can use the indices in $\cI_k$ to transmit information over $W$ for the $k$th decoder
using successive cancellation decoder of polar codes. Moreover,
$$
\lim_{m \to \infty} \frac{1}{N} |\cI_k^{(N)}|=I(U;V_k).
$$
Given an integer $s$, let $l=sN$ and define index sets $\cI_{k}^{(l)}, \cH_{U}^{(l)}$ and $\cL_{U}^{(l)}, 1 \leq k \leq K$ as follows:
$$
\cI_{k}^{(l)}=\{i \in [l]: (i-(\lceil {i}/{N} \rceil -1)N) \in \cI_k^{(N)} \}.
$$
$$
\cH_{U}^{(l)}=\{i \in [l]: (i-(\lceil {i}/{N} \rceil-1)N) \in \cH_U^{(N)} \}.
$$
$$
\cL_{U}^{(l)}=\{i \in [l]: (i-(\lceil {i}/{N} \rceil -1)N) \in \cL_U^{(N)} \}.
$$
It is easy to see that
$$
\frac{1}{l} |\cI_{k}^{(l)}|=\frac{1}{N} |\cI_k^{(N)}|
$$

For every subset $\cA \subset [K], \cA \neq \emptyset$, we define 
$$
\cI_{\cA}^{(l)}=\big( \cap_{k \in \cA} \cI_{k}^{(l)} \big) \cap \big( \cap_{k \in \cA^c} ({\cI_{k}^{(l)}})^c \big).
$$
It is easy to see that
$
\cI_{k}^{(l)}=\bigcup_{\cA \ni k} \cI_{\cA}^{(l)},
$
and the sets $\cI_{\cA}^{(l)}$ are pairwise disjoint for different $\cA$.

Given a rate value $\kappa$ satisfying \eqref{condition}, we can always find an integer $t$ large enough such that 
$\frac{\lceil \kappa t \rceil}{t}$ is as close to $\kappa$ as desired. Let $s=\lceil \kappa t \rceil, l=sN, n=tN$. 
We can also choose a sufficiently large $N$ such that $\frac{1}{l} |\cI_{k}^{(l)}|$ is arbitrarily close to $I(U;V_k)$. 
Thus $\frac{1}{n} |\cI_{k}^{(l)}|$ can be made arbitrarily close to $\kappa I(U;V_k)$. 
By \eqref{condition}, if we choose $t$ and $N$ large enough, then the set of numbers
$\{\frac{1}{n}| \cI_{\cA}^{(l)}|\}_{\cA \subset [K], \cA \neq \emptyset}$ will cover $\{H(X|Y_k)\}_{k=1}^K$ 
in the sense of Def.~\ref{def:cover}. Therefore,
we can use the coding scheme proposed in the proof of Proposition \ref{prop} to find encoders 
$f_{[K]}^{n,\cA}, \cA \subset [K], \cA \neq \emptyset$ and decoders $g_{{[K]}}^{n,k}, 1 \le k \le K$ such that the error probability is 
arbitrarily small and the rate values of the encoders satisfy 
$nR_{{[K]}}^{\cA} <|\cI_{\cA}^{(l)}|$ for all $\cA \subset [K], \cA \neq \emptyset$. 

Given a realization $x^{1:n}$ of the source, the encoder produces a sequence $d^{1:l}$ as follows. 
First, for every $\cA \subset [K], \cA \neq \emptyset$ the coordinates $d^i,i\in \cI_{\cA}^{(l)}$ are filled with the sequence 
$f_{[K]}^{n,\cA}(x^{1:n}).$ 
Since $nR_{{[K]}}^{\cA} <|\cI_{\cA}^{(l)}|$, there will be some extra positions in $\cI_{\cA}^{(l)}$. 
The encoder fills these extra positions with samples of independent uniform binary random variables. 
For $i \in \cL_{U}^{(l)}$, the encoder sets  
  $$
  d^i=\arg\!\!\!\max\limits_{a \in \{0,1\}}\!\!P_{D^i|D^{((\lceil \frac{i}{N} \rceil -1)N+1):(i-1)}}(a|d^{((\lceil \frac{i}{N} \rceil -1)N+1):(i-1)}).
  $$
   For all the other indices, the encoder again sets $d^i$ to be $0$ or $1$ uniformly, independent of everything else. 
   
   In the next step, the encoder calculates 
   \begin{equation}\label{eq:d1}
   u^{1:l}=d^{1:l}\text{diag}(G_N,\dots,G_N)
   \end{equation}
where the diagonal matrix is formed of $s$ identical blocks $G_N$; see \eqref{eq:diag}.
The sequence $u^{1:l}$ is sent through the channel $W$. 
Decoder $k$ can recover the indices in $\cI_{k}^{(l)}=\bigcup_{\cA \ni k} \cI_{\cA}^{(l)}$. Thus decoder $k$ knows 
$f_{ [K]}^{n,\cA}(X^{1:n})$ for all $\cA \ni k$. Together with the realization of $Y_k^{1:n}$, it can now recover $X^{1:n}$.

In the channel transmission part, as indicated in \cite{Honda13}, we need to make the distribution of 
$D^{({\cL_{U}^{(l)}})^c}$ close to i.i.d. uniformly distributed binary random variables. 
Since $\cI_{\cA}^{(l)} \subset (\cL_{U}^{(l)})^c$ for all 
$\cA \subset [K], \cA \neq \emptyset$, we need to make sure that the distribution of the encoded sequence 
$f_{[K]}^{n,\cA}(X^{1:n})$ is close to a uniform distribution on $\{0,1\}^{nR_{{[K]}}^{\cA}}.$ 

Define random variables $C^{1:n}$ by 
  \begin{equation*}
   C^{1:n}=X^{1:n}\text{diag}(G_N,\dots,G_N)
   \end{equation*}
where the block-diagonal matrix on the right is again formed 
as in \eqref{eq:diag}. It can be inferred from the proof of Proposition \ref{prop} that 
$f_{[K]}^{n,\cA}(X^{1:n})$ consists of linear combinations of the bits in 
$C^{\cH_{X}^{(n)}},$ where the set 
   $
   \cH_{X}^{(n)}
  $
is defined as follows:
  $$
  \cH_{X}^{(n)}=\{i \in [n]: (i-(\lceil {i}/{N} \rceil-1)N) \in \cH_X^{(N)} \}.
  $$
 Since the distribution of $C^{\cH_{X}^{(n)}}$ is very close to a uniform distribution, 
we only need to make sure that the linear combinations in the encoded sequence are linearly independent. 
Thus in the encoding procedure, if we find an encoded bit to be a linear combination of the bits in the previously 
encoded sequence, we replace it by a uniform random variable independent of any other random variables. 

The coding scheme described in this section achieves the set of transmission rates in Theorem \ref{thm:T}. The scheme is explicit and relies
on low-complexity encoding and decoding procedures inherited from the basic polar code construction.

\vspace*{.1in}
{\sc Acknowledgment:} We are grateful to our colleague Prakash Narayan for bringing the paper \cite{Tuncel06} to our attention.

\bibliographystyle{IEEEtran}
\bibliography{universal}

\end{document}